\documentclass[10pt,a4paper]{article}
\usepackage{amssymb} 
\usepackage{amsmath} 
\usepackage{graphicx}
\usepackage{slashed}
\usepackage{lipsum}
\usepackage{amsmath}
\usepackage{curve2e}
\usepackage{epsfig,latexsym,color,xcolor}
\usepackage{ulem}
\usepackage{amstext}    
\usepackage{array} 
\usepackage{amsmath}
\usepackage{amsthm}
\newtheorem{theorem}{Theorem}
\begin{document}
\def\Giulia{\bf\color{red}}
\def\bef{\begin{figure}}
\def\eef{\end{figure}}
\newcommand{\ans}{ansatz }
\newcommand{\be}[1]{\begin{equation}\label{#1}}
\newcommand{\beq}{\begin{equation}}
\newcommand{\ee}{\end{equation}}
\newcommand{\beqn}[1]{\begin{eqnarray}\label{#1}}
\newcommand{\eeqn}{\end{eqnarray}}
\newcommand{\bd}{\begin{displaymath}}
\newcommand{\ed}{\end{displaymath}}
\newcommand{\mat}[4]{\left(\begin{array}{cc}{#1}&{#2}\\{#3}&{#4}
\end{array}\right)}
\newcommand{\matr}[9]{\left(\begin{array}{ccc}{#1}&{#2}&{#3}\\
{#4}&{#5}&{#6}\\{#7}&{#8}&{#9}\end{array}\right)}
\newcommand{\matrr}[6]{\left(\begin{array}{cc}{#1}&{#2}\\
{#3}&{#4}\\{#5}&{#6}\end{array}\right)}
\newcommand{\cvb}[3]{#1^{#2}_{#3}}
\def\lsim{\raise0.3ex\hbox{$\;<$\kern-0.75em\raise-1.1ex
e\hbox{$\sim\;$}}}
\def\gsim{\raise0.3ex\hbox{$\;>$\kern-0.75em\raise-1.1ex
\hbox{$\sim\;$}}}
\def\abs#1{\left| #1\right|}
\def\simlt{\mathrel{\lower2.5pt\vbox{\lineskip=0pt\baselineskip=0pt
           \hbox{$<$}\hbox{$\sim$}}}}
\def\simgt{\mathrel{\lower2.5pt\vbox{\lineskip=0pt\baselineskip=0pt
           \hbox{$>$}\hbox{$\sim$}}}}
\def\unity{{\hbox{1\kern-.8mm l}}}
\newcommand{\eps}{\varepsilon}
\def\ep{\epsilon}
\def\ga{\gamma}
\def\Ga{\Gamma}
\def\om{\omega}
\def\omp{{\omega^\prime}}
\def\Om{\Omega}
\def\la{\lambda}
\def\La{\Lambda}
\def\al{\alpha}
\def\beq{\begin{equation}}
\def\eeq{\end{equation}}
\newcommand{\ov}{\overline}
\renewcommand{\to}{\rightarrow}
\renewcommand{\vec}[1]{\mathbf{#1}}
\newcommand{\vect}[1]{\mbox{\boldmath$#1$}}
\def\tm{{\widetilde{m}}}
\def\mcirc{{\stackrel{o}{m}}}
\newcommand{\Dm}{\Delta m}
\newcommand{\dm}{\varepsilon}
\newcommand{\tanb}{\tan\beta}
\newcommand{\nbar}{\tilde{n}}
\newcommand\PM[1]{\begin{pmatrix}#1\end{pmatrix}}
\newcommand{\up}{\uparrow}
\newcommand{\down}{\downarrow}
\newcommand{\refs}[2]{eqs.~(\ref{#1})-(\ref{#2})}
\def\omE{\omega_{\rm Ter}}
\newcommand{\eqn}[1]{eq.~(\ref{#1})}
%

\newcommand{\DSUSY}{{SUSY \hspace{-9.4pt} \slash}\;}
\newcommand{\DCP}{{CP \hspace{-7.4pt} \slash}\;}
\newcommand{\mc}{\mathcal}
\newcommand{\gr}{\mathbf}
\renewcommand{\to}{\rightarrow}
\newcommand{\gtc}{\mathfrak}
\newcommand{\wh}{\widehat}
\newcommand{\br}{\langle}
\newcommand{\kt}{\rangle}


\def\lsim{\mathrel{\mathop  {\hbox{\lower0.5ex\hbox{$\sim$}
\kern-0.8em\lower-0.7ex\hbox{$<$}}}}}
\def\gsim{\mathrel{\mathop  {\hbox{\lower0.5ex\hbox{$\sim$}
\kern-0.8em\lower-0.7ex\hbox{$>$}}}}}

\def\nn{\\  \nonumber}
\def\de{\partial}
\def\brf{{\mathbf f}}
\def\bbf{\bar{\bf f}}
\def\bF{{\bf F}}
\def\bbF{\bar{\bf F}}
\def\bA{{\mathbf A}}
\def\bB{{\mathbf B}}
\def\bG{{\mathbf G}}
\def\bI{{\mathbf I}}
\def\bM{{\mathbf M}}
\def\bY{{\mathbf Y}}
\def\bX{{\mathbf X}}
\def\bS{{\mathbf S}}
\def\bb{{\mathbf b}}
\def\bh{{\mathbf h}}
\def\bg{{\mathbf g}}
\def\bla{{\mathbf \la}}
\def\bmu{\mathbf m }
\def\by{{\mathbf y}}
\def\bmu{\mbox{\boldmath $\mu$} }
\def\bsig{\mbox{\boldmath $\sigma$} }
\def\bunity{{\mathbf 1}}
\def\cA{{\cal A}}
\def\cB{{\cal B}}
\def\cC{{\cal C}}
\def\cD{{\cal D}}
\def\cF{{\cal F}}
\def\cG{{\cal G}}
\def\cH{{\cal H}}
\def\cI{{\cal I}}
\def\cL{{\cal L}}
\def\cN{{\cal N}}
\def\cM{{\cal M}}
\def\cO{{\cal O}}
\def\cR{{\cal R}}
\def\cS{{\cal S}}
\def\cT{{\cal T}}
\def\eV{{\rm eV}}

\numberwithin{equation}{section}

\vspace{6mm}

\large
 \begin{center}
  {\Large \bf  
Generalized Holographic Principle, Gauge Invariance and the Emergence of Gravity \`a la Wilczek}

\end{center}

\vspace{0.1cm}

\begin{center}
{\large Andrea Addazi$\!\!\phantom{a}^{a)\,b)}$, Pisin Chen$\!\!\phantom{a}^{c)\,d)\, e)}$, Filippo Fabrocini$\!\!\phantom{a}^{f)}$, Chris Fields$\!\!\phantom{a}^{g)}$, \\
Enrico Greco$\!\!\phantom{a}^{h)}$, Matteo Lulli$\!\!\phantom{a}^{i)}$, Antonino Marcian\`o$\!\!\phantom{a}^{j)\,k)}$\footnote{E-mail: \, 
marciano@fudan.edu.cn} \& Roman Pasechnik$\!\!\phantom{a}^{l)}$} 
\\
\vspace{0.6cm}
{\it $\phantom{a}^{a)}$ Center for Theoretical Physics, College of Physics Science and Technology, Sichuan University, 610065 Chengdu, China}\\ 
\vspace{0.3cm}
{\it $\phantom{a}^{b)}$ INFN sezione Roma {\it Tor Vergata}, I-00133 Rome, Italy, EU}\\ 
\vspace{0.3cm}
{\it $\phantom{a}^{c)}$ Leung Center for Cosmology and Particle Astrophysics,
National Taiwan University, Taipei, Taiwan 10617}\\
\vspace{0.3cm}
{\it $\phantom{a}^{d)}$ Department of Physics, National Taiwan University, Taipei, Taiwan 10617}\\
\vspace{0.3cm}
{\it $\phantom{a}^{e)}$ Kavli Institute for Particle Astrophysics and Cosmology,
SLAC National Accelerator Laboratory, Stanford University, Stanford, CA 94305, U.S.A.}\\
\vspace{0.3cm}
{\it $\phantom{a}^{f)}$ College of Design \& Innovation, Tongji University, Shanghai, China}\\
\vspace{0.3cm}
{\it $\phantom{a}^{g)}$ 23 Rue de Lavandi\`{e}res, 11160 Caunes Minervois, France, EU\\ 
Tel.: +33-(0)6-44-20-68-69}\\
\vspace{0.3cm}
{\it $\phantom{a}^{h)}$ Institut de Chimie Radicalaire, Aix-Marseille Universit\'e, Marseille, France, EU}\\
\vspace{0.3cm}
{\it $\phantom{a}^{i)}$
Department of Mechanics and Aerospace Engineering, Southern University of Science and Technology, Shenzhen, Guangdong 518055, China}\\
\vspace{0.3cm}
{\it $\phantom{a}^{j)}$ Department of Physics \& Center for Field Theory and Particle Physics,\\ 
Fudan University, 200433 Shanghai, China}\\
\vspace{0.3cm}
{\it $\phantom{a}^{k)}$ Laboratori Nazionali di Frascati INFN, Frascati (Rome), Italy, EU}\\
\vspace{0.3cm}
{\it $\phantom{a}^{l)}$ Department of Astronomy and Theoretical Physics, Lund University, S\"olvegatan 14A
S 223 62 Lund, Sweden, EU}
\end{center}

\vspace{1cm}
\begin{abstract}
\large
\noindent 
We show that a generalized version of the holographic principle can be derived from the Hamiltonian description of information flow within a quantum system that maintains a separable state. We then show that this generalized holographic principle entails a general principle of gauge invariance.
When this is realized in an ambient Lorentzian space-time, gauge invariance under the Poincar\'e group is immediately achieved. We apply this pathway to retrieve the action of gravity. The latter is cast \`a la Wilczek through a similar formulation derived by MacDowell and Mansouri, which involves the representation theory of the Lie groups SO$(3,2)$ and SO$(4,1)$.
\end{abstract}

\baselineskip = 20pt

\section{Introduction}
\noindent 
Almost one hundred years of attempts to quantize gravity suggest that physical perspective may be responsible for this failure \cite{Garay:1994en}. While continuing to seek an UV-complete theory of either General Relativity (GR) or one of its possible extensions \cite{Polchinski:1998rq,Polchinski:1998rr, Rovelli:2004tv, Modesto:2011kw,Modesto:2014lga}, an alternative option is to look at gravity as an emergent phenomenon \cite{Jacobson:1995ab, Barcelo:2005fc, VanRaamsdonk:2010pw, Verlinde:2010hp, Lee:2010bg, Swingle:2014uza, Chiang:2015pmz,Oh:2017pkr}. Among many possible instantiations of this simple idea stands a paradigm of emergence that aims at recovering gravity via its analogical similarity with Yang-Mills gauge theories. 
As remarked by Chen-Ning Yang, while electromagnetism is evidently a gauge theory, and the fact that gravity can be seen as such a theory is universally accepted, how this exactly happens to be the case must be still clarified.  Notable explorations along these lines have been provided in the past by Weyl \cite{weyl}, and more recently by MacDowell and Mansouri \cite{macman}, and Chamseddine and West \cite{chamwest}, with subsequent improvements by Stelle and West \cite{stewe}. 

At the same time, we heuristically note that gravity may naturally encode principles of information theory. Such consideration naturally follows pondering that gravity is the field that is involved in the very definition of both masses and spacetime distances, and that specifies the propagation velocities of point-like particles, and hence of information, through the geodesic equations. Thus it is reasonable to pursue a fundamental theory of gravity from this perspective. Indeed, the underlying graph-structure of information networks is a set of nodes and links --- this is reminiscent of the basis of the states in Loop Quantum Gravity \cite{Rovelli:2004tv}.

There have been huge achievements in the direction of a quantum-information based theory of gravity, with several different attempts developed so far --- see e.g. \cite{Faulkner:2013ica}. More generally, deep links between quantum information theory and an ``emergent'' quantum theory of observable physical systems have been developed by many studies \cite{dariano:16,Hamma:2010sc,hamma}.  It is not within the present scope to summarize this vast literature. Instead, we focus on a specific alternative approach: we show that when the holographic principle is reformulated from a semi-classical to a fully general, quantum-theoretic principle, gravity emerges as a gauge theory along the lines of the gauge formulation of gravity, as proposed by F.~Wilczek in \cite{wilczek}.

We start by showing in Sec.~\ref{2} that a generalized holographic principle (GHP) characterizes information transfer within any finite quantum system in a separable state.  The HP is recovered from this more general, purely-quantum principle by requiring covariance.  We then show in Sec.~\ref{3} that compliance with the GHP entails gauge invariance under the Poincar\'e group in an ambient Lorentzian space-time.  Hence the gauge principle has purely quantum-theoretic roots and characterizes all finite systems in separable states.  
We use this to retrieve the action of gravity in Sec.~\ref{4}. In Sec.~\ref{5}, we provide, as an example, an emergent theory of gravity, a theory of Yang-Mills gauge fields and Higgs pentaplets that is cast \`a la Wilczek. This is a formulation similar to a previous one envisaged by MacDowell and Mansouri, which involves the representation theory of the Lie group SO$(4,1)$, but without explicit symmetry breaking.  We finally summarize some conclusions in Sec.~\ref{6}, and suggest that the AdS/CFT and dS/CFT correspondences may naturally arise within this framework.

\section{Generalized Holographic Principle for finite quantum systems}
\label{2}

\subsection{Historical remarks on the genesis of the Holographic Principle}

Probably the most direct way to summarize the Holographic Principle (HP) is via its original statement by `t~Hooft \cite{tHooft:93}: 

\begin{quote}
{\it ``given any closed surface, we can represent all that happens inside it by degrees of freedom on this surface itself.''}
\end{quote}

\noindent
The path that led to the formulation of the HP can be traced from the Bekenstein's area law \cite{bekenstein:04} for black holes (BH), 
\begin{equation}
\label{Bek}
S = \frac{A}{4},
\end{equation}
where $S$ denotes the thermodynamic entropy of a BH and $A$ its horizon area in Planck units.  Bekenstein conjectured the existence of an upper bound, $S$ itself, to the entropy of any physical system contained within a bounded volume:

\begin{quote}
{\it ``the entropy contained in any spatial region will not exceed the area of the region's boundary.''}
\end{quote}
\noindent
Historically, this conjecture was first instantiated by Susskind \cite{susskind:95}, who implemented a mapping from volume to surface degrees of freedom for a general closed system. This was based on the assumption that all light rays that are normal to any element within the volume are also normal to the surface. Bousso \cite{bousso:02} then showed that it is actually covariance that induces the holographic limit on information transfer by light; he further provided several counterexamples showing the failure of a straightforward interpretation of the HP as a spacelike entropy bound. Instead, Bousso formulated covariant entropy bound: 
\begin{equation}
\label{Bo}
S(L(\Sigma))\leq\frac{A(\Sigma)}{4},
\end{equation} 
with $A(\Sigma)$ denoting the area in Planck units of a (typically but not necessarily \cite{bousso:02}) closed surface $\Sigma$, and $L(\Sigma)$ any light-sheet of $\Sigma$, defined as any collection of converging light rays that propagate from $\Sigma$ toward some focal point away from $\Sigma$. The bound \eqref{Bo} then refers to the entropy of the light-sheet $L(\Sigma)$.  This covariant formulation of the HP holds for the light-sheets of any surface $\Sigma$. BH emerge as special cases, for which the equality in \eqref{Bo} holds. 

We note that both \eqref{Bek} and \eqref{Bo} are semiclassical. The limits on the entropy $S$ that they impose are ``quantum'' only in their reliance on Planck units and hence a finite value of $\hbar$. The entropy itself is classical and of statistical origin, but the finite value of $\hbar$ restricts this thermodynamic entropy within the volume enclosed by $\Sigma$.  In the context of general relativity (GR), $\Sigma$ is a continuous classical manifold enclosing a continuous classical volume characterized by a real-valued metric. As 't~Hooft \cite{tHooft:93} points out, the HP renders  $S(L(\Sigma))$ independent of the metric inside $\Sigma$:

\begin{quote}
{\it ``The inside metric could be so much curved that an entire universe could be squeezed inside our closed surface, regardless how small it is.  Now we see that this possibility will not add to the number of allowed states at all.''}
\end{quote}
\noindent
It bears emphasis that ``allowed states'' in this context are {\it thermodynamic} states, i.e. states that can be counted by measuring energy transfer between the system and its external environment.  As made fully explicit by Rovelli in the case of BH \cite{rovelli:17, rovelli:19}, states that are effectively isolated (e.g. isolated for some time interval much larger than relevant interaction times) from the external environment do not contribute to $S(L(\Sigma))$.

While the demonstration by Maldacena \cite{maldecena:98} of a formal duality acting as an equivalence, at the level of the encoded information, between string quantum gravity on $d$-dimensional anti-de Sitter (AdS) spacetime and conformal quantum field theory (CFT) on its $d-1$-dimensional boundary has made the HP a centerpiece of quantum gravity research, its physical motivation remains that of 't~Hooft's conjecture, namely the inaccessibility of the BH interior summarized by the Bekenstein area law \eqref{Bek}. The HP is conjectured to be fully general, although it is quite mysterious why this should be the case.  As Bousso \cite{bousso:02} remarks, the HP retains a counterintuitive meaning:

\begin{quote}
{\it ``an apparent law of physics that stands by itself, both uncontradicted and unexplained by existing theories that may still prove incorrect or merely accidental, signifying no deeper origin.''}
\end{quote}
\noindent
Our goal in the next section is to place the HP on a much deeper intuitive footing, by generalizing it from a semi-classical to a fully quantum principle, one that is entirely independent of geometric considerations.

\subsection{Information transfer in finite, separable systems}
\label{finite}

Let $\mathbf{S} = AB$ be a closed quantum system characterized by a Hilbert space $\mathcal{H}_\mathbf{S}$ of finite dimension $dim(\mathcal{H}_\mathbf{S})$, and suppose that over some sufficiently long time interval $\bar{\tau}$, $\mathbf{S}$ maintains a separable state, i.e. $|\mathbf{S} \rangle = |AB \rangle = |A \rangle |B \rangle$ for all $t \leq \bar{\tau}$, where $t$ is a time parameter characterizing $\mathbf{S}$.  We can then write a Hamiltonian:

\begin{equation}
H_\mathbf{S} = H_A + H_B + H_{AB}
\end{equation}
\noindent
where $H_{AB}$ is the $A$-$B$ interaction.  Separability allows $H_{AB} = 0$ but requires $H_A, ~H_B \neq 0$.

We now assume $H_{AB} \neq 0$ and choose bases for $A$ and $B$ such that, for all $t \leq \bar{\tau}$:

\begin{equation} 
H_{AB}(t) = \beta^k k_B T^k \sum_i \alpha^k_i(t) M^k_i, \label{ham}
\end{equation}

\noindent
where $k = A$ or $B$, $i = 1 \dots N$ for finite $N \leq dim(\mathcal{H}_S)$, the $\alpha^k_i(t)$ are real functions with codomains $[0,1]$ such that:

\begin{equation}
\sum_i \int_{\Delta t} dt ~\alpha^k_i(t) = \Delta t \label{norm}
\end{equation}
\noindent
for every finite $\Delta t$, $k_B$ is Boltzmann's constant, $T^k$ is $k$'s temperature, $\beta^k \geq$ ln 2 is an inverse measure of $k$'s average per-bit thermodynamic efficiency\footnote{Here the efficiency relates to the thermodynamic transformations triggering the exchange of information bits among the two subsystems $A$ and $B$.} that depends on the internal dynamics $H_k$, and the $M^k_i$ are Hermitian operators with binary eigenvalues.  Given separability, we can interpret these $M^k_i$ as ``measurement'' operators that each transfer 1 bit between $A$ and $B$.  Here the condition $\beta^k \geq$ ln 2 assures compliance with Landauer's principle \cite{landauer:61}: each bit transferred from $A$ to $B$ by the action of some operator $M^A_i$ is paid for by the transfer of an energy $\beta^B k_B T^B$ from $B$ to $A$ and vice-versa.  ``Irreversible recording'' of the transferred bits by $A$ and $B$ corresponds\footnote{Notice that irreversibility is connected to the efficiency bound $1/\beta^k<1$.} to state changes:

\begin{equation}
|A \rangle|_t \rightarrow |A \rangle|_{t + \Delta t} \quad \mathrm{and} \quad |B \rangle|_t \rightarrow |B \rangle|_{t + \Delta t} \label{trans}
\end{equation}
\noindent
that maintain the separability of $\mathbf{S}$.  Given \eqref{norm}, the action required for $k$ to transfer $N$ bits in time $\Delta t$ is:

\begin{equation}
\int_{\Delta t} dt ~(\imath \hbar) \mathrm{ln} \mathcal{P}(t) = N \beta^k k_B T^k \Delta t \label{action}
\end{equation}
\noindent
where $\mathcal{P}(t) = \exp(-(\imath /\hbar) H_{AB} t)$.  Informational symmetry clearly requires $\beta^A T^A = \beta^B T^B$ during any finite $\Delta t$.

Let us now consider an interval $\tau \ll \bar{\tau}$ during which $A$ and $B$ exchange exactly $N$ bits.  In any such interval, the thermodynamic entropy $S(B)|_\tau$ measured by $A$ is clearly $N$ bits; the entropy $S(A)|_\tau$ measured by $B$ is similarly $N$ bits.  Coarse-graining time to an interval $\Delta t = n \tau \ll \bar{\tau}$ to allow $n$ $N$-bit measurements, both measured entropies remain $N$ bits.  Hence we have:

\begin{theorem}
Given any finite-dimensional quantum system $\mathbf{S} = AB$ that maintains a separable state $|AB \rangle = |A \rangle |B \rangle$ for $t \leq \bar{\tau}$, the information $S(B)$ obtainable by $A$ during any finite interval $\Delta t \ll \bar{\tau}$ is independent of $H_B$.
\end{theorem}

\begin{proof}
The information $S(B)$ obtainable by $A$ during any finite interval $\Delta t \ll \bar{\tau}$ is just the information transferred by $H_{AB}$, which is specified entirely independently of $H_B$.  Indeed $H_B$ and hence $H_\mathbf{S}$ can be varied arbitrarily, provided that $B$ has sufficient degrees of freedom to maintain the separability of $\mathbf{S}$, without affecting $H_{AB}$.  
\end{proof}
\noindent
Note that if the assumption of separability is dropped and two subsystems cannot be distinguished, Eq. \eqref{trans} fails, the von Neumann entropy of $|AB \rangle$ remains zero, and no information is transferred by $H_{AB}$.

\subsection{The HP is a special case of the GHP}

Theorem 1 places a principled restriction on information transfer {\it within} any separable quantum system; as noted above, the notion of information transfer {\it within} a non-separable (i.e. entangled) quantum system is meaningless.  The HP is a principled restriction on information transfer {\it within} a semiclassical system that is separable by definition.  Hence the two should be related.  This relation can be made explicit by stating:

\begin{quote}
\noindent
\textbf{Generalized Holographic Principle} (GHP): Given any finite-dimensional quantum system $\mathbf{S} = AB$ meeting the conditions of Theorem 1, the thermodynamic entropies of $A$ and $B$ over a coarse-grained time, over which $A$ and $B$ only interact through Eq.~\eqref{ham}, are $S(B) = S(A) = N$ bits, where $N$ is the number of operators in the representation \eqref{ham} of $H_{AB}$.
\end{quote}
\noindent
We note that this GHP is formulated entirely independently of geometric assumptions; in particular, it is prior to any assumption of general covariance.

To make the physical meaning of the GHP clear, let us consider a specific example.  Suppose $A$ and $B$ interact by alternately preparing and measuring the states of $N$ shared, non-interacting qubits as shown in Fig. 1.  We can consider that, in a time interval $\tau$, $A$ prepares the $N$ qubits in her choice of basis, i.e. using her $M^A_i$ and then $B$ makes measurements in his choice of basis, i.e. using his $M^B_j$.  In the next interval $\tau^{\prime}$, $B$ prepares and $A$ measures, and so forth.  The prepared and measured bit values will be preserved, i.e. $A$ and $B$ will employ the same ``language,'' only if they share a basis, in this case a $z$ axis, which functions as a shared quantum reference frame (QRF) \cite{bartlett:07, fm:19}.  Clearly, however, $S(B) = S(A) = N$ bits in every interval of length at least $2 \tau$, independently of whether $A$ and $B$ share a QRF.

\begin{figure}
\centering
\includegraphics[width=15 cm]{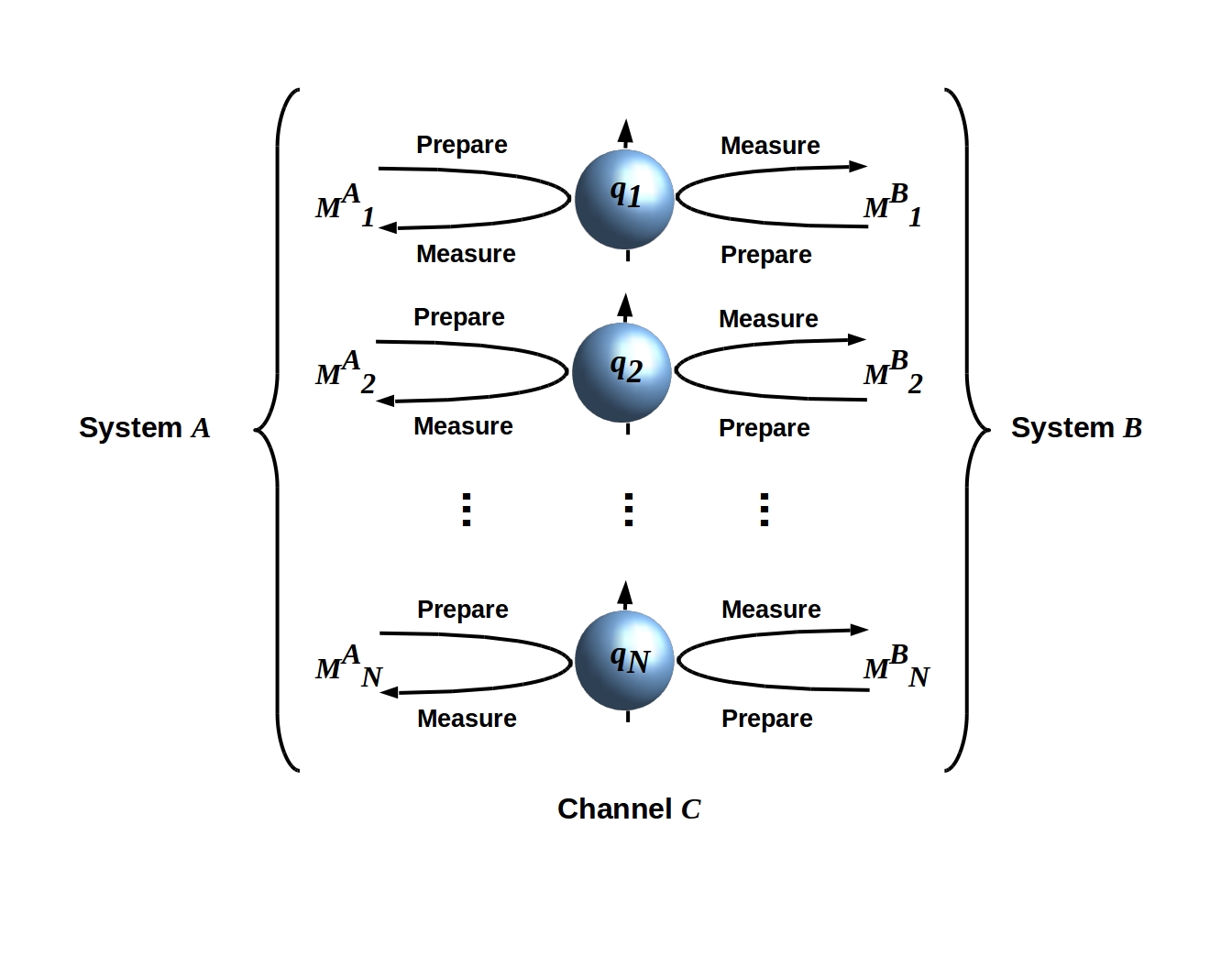}
\caption{Systems $A$ and $B$ exchange bits via an ancillary array of non-interacting qubits.  Bit values are preserved if a quantum reference frame (here, a $z$ axis) is shared \textit{a priori}.}
\end{figure}

The qubit-mediated interaction shown in Fig. 1 still makes no geometric assumptions.  If we now imagine, however, that the array of qubits is embedded at maximal density in an ancillary real 2-dimensional surface $\Sigma$, and further require that the bit values generated by the actions of the $M^A_i$ (respectively, $M^B_j$) must be transferred to a distant point within $A$ (respectively, $B$) by photons (or any other quantum carrier consistent with the local symmetry/invariance that is present), Eq. \eqref{Bek} and \eqref{Bo}, i.e. the usual covariant HP, results by the reasoning of \cite{bekenstein:04, bousso:02}.  The surface $\Sigma$ can, in this case, naturally be interpreted as a ``boundary'' between $A$ and $B$ at which they interact.  The self-interactions $H_A$ and $H_B$ are then naturally interpreted as characterizing the ``bulk'' of $A$ and $B$, respectively.

We note that the GHP provides, when $H_{AB}$ is assumed to act across an $A$-$B$ boundary, an immediate and intuitive explanation of the decoherence of $B$ relative to $A$ and vice-versa \cite{fm:19, fields:19}.  Hence the GHP provides a natural account of the ``emergence of classicality'' {\it within} separable quantum systems: if $|AB \rangle$ is separable as $|A \rangle |B \rangle$, ``classicality'' characterizes the bit values ``encoded on'' the $A$-$B$ boundary, i.e. the boundary at which $H_{AB}$ acts.  There are, in other words, no classical systems, just classical information.

Both Theorem 1 and the GHP above are formulated for fixed $N$.  Generalizing to the case of $N$ varying slowly, i.e. remaining piecewise constant in time for intervals $\tau \ll \Delta t \ll \bar{\tau}$, is straightforward.  Therefore, only the constant $N$ case is needed in what follows.

\section{Poincar\'e symmetries and gauge invariance}
\label{3}

\subsection{The GHP requires gauge invariance for finite, separable systems}

Theorem 1 and hence the GHP restricts access to information, and so states an invariance: the information $S(B)$ is invariant under changes in $H_B$ (and vice-versa), provided $H_{AB}$ remains fixed and separability is maintained.  Gauge invariance for the ``bulk'' Hamiltonians $H_A$ and $H_B$ clearly follows.

\begin{theorem}
In any $\mathbf{S} = AB$ compliant with the GHP, the bulk interactions $H_A$ and $H_B$ are gauge invariant.
\end{theorem}

\begin{proof}
The situation is completely symmetrical, so considering either $H_A$ or $H_B$ alone is sufficient.  Gauge symmetry for $H_A$ can only fail if a local coordinate change, i.e. a local change in basis vectors for $\mathcal{H}_A$, is observable, i.e. has an effect on $H_{AB}$.  Any such effect is ruled out by Theorem 1, which is satisfied by all systems compliant with the GHP.
\end{proof}
\noindent
Note that gauge invariance here depends explictly on separability: if $|AB \rangle \neq |A \rangle |B \rangle$, i.e. $\mathbf{S}$ is in an entangled state, the notion of a ``bulk'' Hamiltonian $H_A$ (or $H_B$) is meaningless.

Theorem 2, like Theorem 1 and the GHP, involves no assumptions about geometry.  We introduce these below, with QED as an initial example.

\subsection{QED and the consequences of the GHP}

As a specific example, consider a finite system $AB$, with $A$ comprising photons described by the usual electromagnetic vector field $A_\mu(x)$ and $B$ comprising fermionic particles, e.g. electrons described by the Dirac field $\psi(x)$, with $x$ a real space-time coordinate.  Clearly $A$ and $B$ only interact via a Hamiltonian $H_{AB}$.  The numbers of photons and electrons can be arbitrarily large, so the usual approximation of infinite degrees of freedom can be adopted for simplicity with no physical (i.e. observable) consequences.


We can make the presentation more precise at the mathematical level to illustrate the independence of observable results from coordinate (i.e. basis) transformations, even in the presence of the ancillary space-time geometry with points labelled by $x$. The local gauge freedom for the choice of the vector field $A_\mu$, which generalizes for quantum fields, i.e. systems with infinite degrees of freedom embedded in an ancillary space $x$, the invariance with respect to the choice of basis we discussed above, can be defined as in: 
\begin{equation} \label{g1}
A_\mu(x)\rightarrow A_\mu'(x)=A_\mu(x) -\partial_\mu \lambda(x)\,,
\end{equation}
where $\lambda(x)$ is a scalar function that is continuous with its first derivatives. Stepping out of the redundancy contained in \eqref{g1}, and denoting the components of the vector fields as $A_\mu=\{\phi, \vec{A}\}$, the electric and the magnetic fields, invariant under  \eqref{g1}, can be defined respectively as $\vec{E}=\dot{\vec{A}} - \nabla{\phi}$ and $\vec{B}=\nabla\wedge\vec{A}$. 

The quantization of the theory can be achieved following the standard path integral procedure. The partition function for the system $A$, namely the U$(1)$ gauge sector, casts:  
\begin{equation} \label{path}
\mathcal{Z}_A[A]= \int \mathcal{D} A_\mu \, e^{\imath \mathcal{S}(A)}\,,
\end{equation}
in which $\mathcal{D} A_\mu$ denotes the path integral measure over the copies of the gauge field, while $\mathcal{S}(A)$ denotes the classical action. The expectation value in the path integral of the theory of any functional observable $\mathcal{O}[A]$ is invariant under the gauge transformation \eqref{g1}. It is straightforward to show this fundamental property by comparing the expectation value of $\mathcal{O}[A]$ for different choices of the gauge fixing. On the other hand, different choices of the gauge fixing correspond to different choices of (local) observers, namely bases. But within these circumstances, the GHP implies that the path integral of any $\mathcal{O}[A]$ must be gauge invariant: choosing an observer is choosing a measurement basis, i.e. choosing a set of operators $M^k_i$ in \eqref{ham}, a choice that is independent of $A$ by Theorem 1. In simpler words, the {\it GHP implies gauge invariance}. 

The redundancy due to the gauge transformations, i.e. choices of $M^k_i$, can be factored out by fixing the gauge functional $G$, and then imposing gauge invariance. This is achieved by inserting in the path integral \eqref{path} the resolution of the identity:
\begin{equation} \label{gpath}
1= \int \mathcal{D} \lambda \delta (G(A^\lambda)) \, \left| {\rm det} \frac{\delta G(A^\lambda)}{\delta \lambda (x)}\right|\,,
\end{equation}
where as customary  $A^\lambda_\mu(x)=A_\mu(x)+\partial_\mu \lambda(x)$. The simplest choice of gauge functional is provided by the Lorentz functional $G(A)=\partial_\mu A^\mu$, which implements the gauge invariance of the path integral. Notice that under gauge transformations, the Lorentz functional transforms as: 
\begin{equation}
G(A)=\partial_\mu A^\mu \quad \rightarrow \quad G(A^\lambda)=\partial_\mu A^\mu +\Box\lambda\,.
\end{equation}
Having all set up, we can easily show the invariance of the path-integral:
\begin{eqnarray} 
\mathcal{Z}_{\mathbf{A}}[A]&&=\mathcal{N}\, \int \, \mathcal{D} A_\mu \, e^{\imath \mathcal{S}_{\rm A}(A)}\, \int \mathcal{D} \lambda \delta (G(A^\lambda)) \, \left| {\rm det} \frac{\delta G(A^\lambda)}{\delta \lambda (x)}\right|  \nonumber\\
&&=\mathcal{N} \left| {\det \Box}\right| \,\int\,  \mathcal{D} A_\mu \, e^{\imath \mathcal{S}_{\rm A}(A)}\,\int \mathcal{D} \lambda \delta (G(A^\lambda)) \nonumber\\
&&= \mathcal{N}'\, \,\int \mathcal{D} \lambda \, \mathcal{D} A_\mu \, e^{\imath \mathcal{S}_{\rm A}(A)} \, \delta (G(A))  \nonumber\\
&&= \mathcal{N}'' \, \int \, \mathcal{D} A_\mu \, e^{\imath \mathcal{S}_{\rm A}(A)} \, \delta (G(A))\,, \label{sassa}
\end{eqnarray} 
where the normalization functions $\mathcal{N}$, $\mathcal{N}'$ and  $\mathcal{N}''$ are not relevant, and have been safely disregarded. \\

The perspective provided by the GHP allows us to place a novel physical interpretation on the action of $G(A)$ in removing gauge redundancy, one that points toward a deeper understanding of role of the spatial coordinate $x$ in QFT.  As noted above, from a GHP perspective, gauge redundancy is the redundancy in the choice of measurement operators $M^k_i$.  This can equally be interpreted as redundancy in the choice of observers $k$.  But $k$, in this case, is just a quantum system $X$ that can be coupled to the quantum field $A$ while maintaining a separable joint state $|AX \rangle$.  The action of $G(A)$ renders these different observers redundant, effectively removing the dependence of (observations of) $A$ on $x$.  Hence we can now see what $x$ is doing in QFT: it is enforcing separability.  This is indeed an insight of Einstein \cite{einstein:48}:

\begin{quote}
Further, it appears to be essential for this arrangement of the things introduced in physics that, at a specific time, these things claim an existence independent of one another, insofar as these things ``lie in different parts of space.''
\end{quote}
\noindent
``Claiming an existence independent of one another'' obviously requires separability.\\

It is well known that the gauge condition can be cast in a more general form, employing an arbitrary function $f$. In this latter case, the gauge functional:
\begin{equation}
G_f(A)=\partial_\mu A^\mu -f 
\end{equation}
actually introduces a family of gauge-fixing terms. The independence of the physical observables on the gauge fixing is recovered through a process of average that is realized by integrating over $f$ the gauge fixing terms weighted with the factor $\exp (\, {-\frac{\imath}{2\xi} \int d^4 x f^2(x)})$, where $\xi$ is a positive parameter. The path integral in \eqref{sassa} then recasts as: 
\begin{eqnarray}
\mathcal{Z}_{\mathbf{A}}[A]&&=\mathcal{N} \, \int  \,\mathcal{D} f \,\mathcal{D} A_\mu \, e^{\imath \mathcal{S}_{\rm A}(A)- \frac{\imath}{2\xi} \int d^4 x \, f^2(x)} \nonumber\\
&&=\mathcal{N} \, \int \,\mathcal{D} A_\mu \, e^{\imath \mathcal{S}_{\rm A}(A)- \frac{\imath}{2\xi} \int d^4 x (\partial_\mu A^\mu(x))^2}
\,.
\end{eqnarray}
This invariance of the partition function under different choices of the gauge fixing condition, i.e. different choices of $f$, percolates into the gauge invariance of the expectation value of any observable $\mathcal{O}$. This manifests immediately, as one can recognize from the easy passages 
\begin{eqnarray}
\langle \mathcal{O}[A] \rangle_f&&=\mathcal{N} \, \int \,\mathcal{D} f \,\mathcal{D} A_\mu \, e^{\imath \mathcal{S}_{\rm A}(A)- \frac{\imath}{2\xi} \int d^4 x \, f^2(x)} \, \mathcal{O}[A]\nonumber\\
&&=\mathcal{N} \, \int \,\mathcal{D} A_\mu \, e^{\imath \mathcal{S}_{\rm A}(A)- \frac{\imath}{2\xi} \int d^4 x (\partial_\mu A^\mu(x))^2} \, \mathcal{O}[A]\nonumber\\
&&=  \mathcal{N} \, \int \,\mathcal{D} g \,\mathcal{D} A_\mu \, e^{\imath \mathcal{S}_{\rm A}(A)- \frac{\imath}{2\xi} \int d^4 x \, g^2(x)} \, \mathcal{O}[A] \nonumber\\
&&=\langle \mathcal{O}[A] \rangle_g\,,
\end{eqnarray}
where $f$ and $g$ are two different gauge-fixings. 

We may take into account now the other interaction partner in QED, the system $B$ composed by Dirac matter fields, for simplicity electrons. The path integral formulation of the system, composed by only one fermionic species $\psi$, then casts :
\begin{equation} \label{pathD}
\mathcal{Z}_{B}[\psi, \bar{\psi}]= \int \mathcal{D} \psi\,  \mathcal{D} \bar{\psi} \, e^{\imath \mathcal{S}_{\rm D}(\psi, \bar{\psi})}\,,
\end{equation}
where
\begin{equation} \label{SD}
\mathcal{S}_{\rm D}(\psi, \bar{\psi})=\int d^4x\, \bar{\psi} \left( \imath \gamma^\mu \partial_\mu -m\right)\psi\,.
\,
\end{equation}

The observable quantities $\mathcal{O}[\psi,\bar{\psi}]$ are bilinear operators in the fermionic fields $\psi$ and $\bar{\psi}$, which can be generally cast as $\mathcal{O}[\psi,\bar{\psi}]=\bar{\psi} \, O(\Gamma_I) \psi$, where the matrix $O$, with suppressed spinorial indices, depends on the elements of the Clifford algebra $\Gamma_I$, with $I=1\dots16$. 

As previously done for the bosonic system $A$, also for the system $B$ we can introduce a local gauge transformation having the meaning of a transformation among observers. Of course, this transformation cannot change the values of $\mathcal{O}$, which lead us to state the necessity of the symmetry 
\begin{equation} \label{gaugepsi}
\psi(x)\rightarrow  \psi'(x)=e^{\imath q \lambda(x)} \psi(x)\,, \qquad \bar{\psi}(x)\rightarrow  \bar{\psi}'(x)=e^{-\imath q \lambda(x)} \bar{\psi}(x) \,,
\end{equation}
where $q$ stands for the charge parameter. This is a U$(1)$ transformation, the generator of which commutes with the matrix $O$, ensuring the gauge invariance of any observable $\mathcal{O}$ under
\eqref{gaugepsi}. By the Noether theorem, selecting $\lambda(x)=\alpha\in \mathbb{R}$ to individuate an infinitesimal global transformation, the conserved charge is easily recovered 
$$Q=\int_\Sigma d^3x\, \psi^\dagger \psi\,,$$ 
where $\Sigma$ is a spatial hypersurface. This allows to cast U$(1)$ transformations acting on the Hilbert space of the theory as $U=e^{\imath \alpha Q}$.

\subsection{Extension to gravity and local Lorentz invariance}

So far we have first considered generic quantum systems with finite number of degrees of freedom, and stated the GHP within these simplified but completely general contexts, which do not necessarily require geometric concepts. In this sense, these notions shall be considered as pre-geometric. We have then extended our focus to continuous systems with an infinite number of degrees of freedom, focusing specifically on the paradigmatic example of QED, which is embedded on a flat Minkowski space-time. This embedding requires the addition of ancillary coordinates $x$ into the description of the system, which are necessary to specify its evolution and fully capture the dynamics as it is observed by spatially-separated observers. \\

Let us now include gravity in this construction, extending the arguments previously exposed. Our joint system $\mathbf{S} = AB$ shall be now composed by the gravitational degrees of freedom, described by the gravitational field $g_{\mu\nu}(x)$, the configurational space of which constitutes the system $A$, and by the matter degrees of freedom, the fields\footnote{Here we denote with $\phi^{\aleph}(x)$ any possible scalar, vectorial or spinorial matter fields.} $\phi^{\aleph}(x)$, the configurational space of which individuates the system $B$. The GHP then imposes the gauge invariance of the gravitational field, once the ancillary spatial coordinates have been introduced, in exactly the same way as discussed above. The role of the spatial coordinates is, as in the case of QED, labeling the manifold in which separable physical (in this case, matter) systems, e.g. observers, that interact with the field $A$ are embedded.  Symmetries fully depend in this picture on the signature of the embedding space, which we assume to be Lorentzian, so as to distinguish among time (required already by \eqref{ham}) and space (ancillary) coordinates. Thus the emergent symmetry will impose the invariance under local SO$(3,1)$ transformations: indeed, the underlying space structure we are considering has Lorentzian signature. This is also consistent with the fact that the tangent space to each point of the manifold is flat and Minkowski, and thus the whole construction specifies how the invariance under supertranslations in time and space emerges in this framework. \\

Besides local SO$(3,1)$ Lorentz symmetries, the theory of gravity also encodes symmetries under diffeomorphism. It is customary to deal with these latter in the Hamiltonian formulation of the theory.  This requires considering Lorentzian manifolds $\mathcal{M}_4$ that are diffeomorphic to $\mathbb{R}\times \Sigma$. This property enables a slicing of $\mathcal{M}_4$ into space-like hypersurfaces $\Sigma$ at instants of time $t\in \mathbb{R}$. This slicing is arbitrary, since there exist several possible choices to pick a diffeomorphism $\phi: \mathcal{M}_4 \rightarrow \mathbb{R}\times \Sigma$. Thus different time coordinates $\tau$ can be defined on $\mathcal{M}_4$, as the pullback of $t$ that is realized by $\phi$, or in the mathematical jargon $\tau=\phi*t$. This corresponds to different clocks for different observers (in the language of \S\ref{finite}, different ``tick'' intervals $\tau$), which nevertheless must not affect the definition of the physical observables, which are gauge invariant and diffeomorphic invariant. The theory is then recast on Cauchy surfaces, on which the gravitational field is captured by the restriction to the gravitational field to the three-metric on the slice $\Sigma$, namely $\!\!\!\phantom{a}^{3}g$, and to its ``time'' derivative, namely the extrinsic curvature $K$. These variables form the Cauchy data of the problem, and open the pathway to access the meaning of the ten components of the Einstein equations. Indeed, four of the Einstein equations turn out to be constraint equations that the Cauchy data must satisfy, while the other six are evolutionary equations that dictate the dynamics in time of the three-metric. \\   

A time-like unit vector $n$ must be then defined that is orthogonal to any tangent vector $v$ on $\Sigma$. Considering the metric two-form $g$ on $\mathcal{M}_4$, these two requirements amount to write in formulas that $g(n,n)=-1$ and $g(n,v)=0$. The direction of the unit vector $n$ is then defined to point towards the future. A derivative of any generic vector $v$ on $\Sigma$ can be then defined along any generic direction individuated by a vector $u$ on $\Sigma$. This is simply attained by projecting on $\Sigma$ and removing the component along the normal direction, i.e. $\nabla_u v= -g(\nabla_u v,n)n+ (\nabla_u v +g(\nabla_u v,n))$. The first term identifies the extrinsic curvature, $K(u,v) n= -g(\nabla_u v,n)n$, which measures the bending of the surface $\Sigma$ in the ambient manifold $\mathcal{M}_4$, by quantifying the failure of a generic vector of $\mathcal{M}_4$ to be still tangent to $\Sigma$ after we parallel translate it using the Levi-Civita connection on the ambient space $\mathcal{M}_4$. While considering the component tangential on $\Sigma$, we can write $\phantom{a}^{3}\nabla_u v= \nabla_u v +g(\nabla_u v,n)$, since it introduces the Levi-Civita connection on $\Sigma$ that is associated to the three metric $\!\!\!\phantom{a}^{3}g$. One can show that this is a connection, and satisfies the Leibnitz rule.\\

Given this framework, denoted as ADM \cite{ADM} in the literature, we can now introduce a time coordinate $\tau=\phi^* t$ on $\mathcal{M}_4$, which individuates a generic foliation $\{\tau=s\}$. The vector field $\partial_\tau$ on $\mathcal{M}_4$ then admits the generic decomposition along a tangential direction to $\Sigma$ and its normal, respectively individuated by the lapse function $N$ and the shift vector $\vec{N}$, i.e. $\partial_\tau N n + \vec{N}$. We can finally move to consider the Einstein theory of gravity and its Hamiltonian structure, which is a purely constrained system. We can first move from the Einstein-Hilbert action, which we review in Sec.~\ref{4}, and cast it in terms of the ADM variables, using the three-metric $\!\!\!\phantom{a}^{3}g$, the Levi-Civita connection on $\Sigma$, namely $\!\!\!\phantom{a}^{3}\Gamma$, the associated Riemann tensor $\!\!\!\phantom{a}^{3}R^a_{\ bcd}$ and the Ricci scalar $\!\!\!\phantom{a}^{3}R$ on $\Sigma$. Then the Lagrangian of the Einstein-Hilbert action reads, modulo a boundary term:
\begin{equation}
\mathcal{L}= \sqrt{\!\!\!\phantom{a}^{3}g}\,  N \left( \!\!\!\phantom{a}^{3}R +  {\rm tr}(K^2) - ({\rm tr}K)^2 \right)\,,
\end{equation}
which allows to define the symplectic structure of the system, namely: 
\begin{equation} \label{qpref}
q_{ij}=\!\!\!\phantom{a}^{3}g_{ij}\,, \qquad \qquad p^{ij}=\frac{\delta \mathcal{L}}{\delta \dot{q}_{ij}}=  \sqrt{q} \left( K^{ij} - ({\rm tr}K) q^{ij} \right)\,,
\end{equation}
with vanishing momenta conjugated to $\vec{N}$ and $N$, respectively $\vec{P}=0$ and $P=0$. 
In Eq.~\eqref{qpref} the extrinsic curvature is expressed in terms of the covariant derivatives on $\Sigma$, namely $\phantom{a}^{3}\nabla$, and the ADM variables, as:
\begin{equation}
K_{ij}=\frac{1}{2N} \left(  \dot{q}_{ij} -\!\! \phantom{a}^{3}\nabla_i N_j - \!\! \phantom{a}^{3}\nabla_j N_i  \right)
\,.
\end{equation}
The Hamiltonian density of the gravitational system, which can be calculated by the usual Legendre transform $\mathcal{H}(q_{ij}, p^{ij})=p^{ij} \dot{q}_{ij} - \mathcal{L}$, finally provides the Hamiltonian of the system $H=\int_\Sigma \mathcal{H} d^3x$. This latter can be recognized to be a totally constrained system:  
\begin{equation}
\mathcal{H} = \sqrt{q} \left( N  \mathcal{C} + N^i  \mathcal{C}_i \right)\,,
\end{equation}
with
\begin{equation}
\mathcal{C}= - \!\!\phantom{a}^{3}R +  q^{-1} \left( {\rm tr}(p^2) - ({\rm tr}\, p)^2 \right)
\,,
\qquad
\mathcal{C}_i= -2 \, \phantom{a}^{3}\nabla^j ( q^{-1/2} p_{ij})\,. 
\end{equation}
For simplicity, we assumed the hyperspace $\Sigma$ to be compact, so as to neglect the contribution otherwise provided by total divergences.\\ 

The first term of the Hamiltonian represents the Hamiltonian constraint, which generalizes time reparametrization, while the second term is the space-diffeomorphism constraint, respectively:   
\begin{equation}
C(N)=\int_\Sigma N \mathcal{C} \sqrt{q} d^3x\,, \qquad 
C(\vec{N})=\int_\Sigma N^i \mathcal{C}_i \sqrt{q} d^3x\,.
\end{equation}
Involving the continuous version of the Poisson brackets for the phase-space variables of the system, namely:
\begin{equation}\nonumber
\left\{ f,\,g  \right\} =\int_\Sigma \, \left\{ \frac{\partial f}{\partial p_{ij}(x) } \frac{\partial g}{\partial q^{ij}(x) } 
- \frac{\partial f}{\partial q^{ij}(x)} \frac{\partial g}{\partial p_{ij}(x)}
\right\} \sqrt{q} \,d^3 x\,,
\end{equation}
we may recover the algebra of constraints for the gravitational system, known as Dirac algebra, namely: 
\begin{eqnarray}
\left\{ C(\vec{N}),\, C(\vec{N}')  \right\}\!\!&=&\!\!C([\vec{N}, \vec{N}'])\,, \\
\left\{ C(\vec{N}),\, C( {N})  \right\}\!\!&=&\!\!C(\vec{N}\, N') \,, \\
\left\{ C(N),\, C(N')  \right\}\!\!&=&\!\!C( (N\partial^i N' - N' \partial^i N  )\partial_i ) \,.
\end{eqnarray}
The scalar and the vector constraints entering the total Hamiltonian can be cast in terms of the Einstein tensor components, contracted with the normal $n^\mu$ to the hypersurface $\Sigma$, i.e.:
\begin{eqnarray}\nonumber
\mathcal{C}=-2 G_{\mu \nu } n^\mu n^\nu\,, \qquad \qquad \mathcal{C}_i = -2 G_{\mu i } n^\mu \,,
\end{eqnarray}
while the spatial components $G_{ij}$ source the Hamilton equation of the phase-space variable, which are also captured by the Hamilton equation $\dot{q}^{ij}=\{ H, q^{ij} \}$ and $\dot{p}^{ij}=\{ H, p^{ij} \}$.

\section{Emergent Poincar\'e symmetries from an emergent gauge theory}
\label{4}

There is a deep similarity among gauge symmetries and diffeomorphisms, which becomes manifest as soon as both the gauge theories and gravity are formulated as principle bundle theories. This turns space-time symmetry into an emergent concept, similarly to what has been discussed in the previous section, while considering the consequences of the GHP. A celebrated framework in which gravity, and thus the Poincar\'e symmetries, are shown to be emergent from a gauge structure was provided by MacDowell and Mansouri. Nonetheless, the gauge symmetry is explicitly broken in this model. We briefly review here this theoretical framework, as a propedeutic element to the next section, where we review a model, due by Wilczek, in which gravity is emergent from a fully gauge-invariant theory.

\subsection{Einstein-Hilbert action}

Before introducing MacDowell-Mansouri gravity, it is useful to remind the Palatini formulation of gravity in the Einstein-Hilbert action. This casts in terms of the metric $g_{\mu\nu}$, its inverse $g^{\mu\nu}$, and its first and second derivatives, i.e.: 
\begin{equation}
S_{\rm EH}=\frac{1}{16 \pi G} \int d^4x \sqrt{-g} (R -  2 \Lambda)\,,
\end{equation}
with $R$ being the Ricci scalar. The Ricci scalar, encoding non-linearly first order derivatives and linearly second order ones, is defined as the contraction of the Riemann tensor, namely: 
\begin{equation}
R_{\alpha \beta \mu \nu} g^{\alpha \nu } g^{\beta \mu }=R\,,
\end{equation}
with the Riemann tensor expressed as:
\begin{equation}
R^\alpha_{ \ \beta \mu \nu}=\partial_\mu \Gamma^\rho_{\nu \sigma} - \partial_\nu \Gamma^\rho_{\mu \sigma} + \Gamma^\alpha_{\mu\lambda}\, \Gamma^\lambda_{\nu\beta} -  \Gamma^\alpha_{\nu\lambda}\, \Gamma^\lambda_{\mu\beta}  \,,
\end{equation}
with $\Gamma^\rho_{\mu\nu}$ Christoffel symbols. These latter are torsionless in the Einstein-Hilbert theory, namely $T^\rho_{\mu\nu}=\Gamma^\rho_{\mu\nu}-\Gamma^\rho_{\nu \mu}=0$. Varying with respect to the Christoffel symbol, one obtains the expression in terms of the metric and its derivatives: 
\begin{equation}
\Gamma^\rho_{\mu \nu }=\frac{1}{2} g^{\rho \sigma} (\partial_\mu g_{\nu \sigma} + \partial_\nu g_{\mu \sigma} - \partial_\sigma g_{\mu \nu} )\,.
\end{equation} 

A first-order formulation of the Einstein-Hilbert action of gravity is admitted in terms of the $SO(3,1)$-connection $\omega_\mu^{ab}$ and the tetrad one-form (frame-field) $e^a_\mu$, which is valued in the $\mathfrak{so}(3,1)$ algebra and carries an internal vector-index in the fundamental representation of $\mathfrak{so}(3,1)$. In terms of these fields, the action now casts: 
\begin{equation} \label{topological}
S_{\rm EH}=\frac{1}{64 \pi G} \int d^4x \, \epsilon_{abcd} \, \left( R_{\mu \nu}^{ \ \  ab} \, e_{\rho}^{ \ c} \,  e_{\sigma}^{ \ d} - \frac{\Lambda}{3} \, e_{\mu}^{ \  a} \,  e_{\nu}^{ \ b} \, e_{\rho}^{ \ c} \,  e_{\sigma}^{ \ d} \right) \epsilon^{\mu \nu \rho \sigma}
\,,
\end{equation} 
with:
\begin{eqnarray}
R_{\mu \nu}^{\ \ \ ab} \!&=&\! \partial_\mu \omega_\nu^{ab} - \partial_\nu \omega_\mu^{ab} +  \omega_{\mu \ c}^{\ a} \, \omega_{\nu}^{cb} -  \omega_{\nu \ c}^{\ a}  \, \omega_{\mu}^{cb}  
\,, \nonumber\\
T_{\mu \nu}^a \!&=&\! \mathcal{D}^\omega_\mu \, e_\nu^{\ a} -\mathcal{D}^\omega_\nu \, e_\mu^{\ a}\,,
\end{eqnarray}
which can be recast as $R^{ab}=d\omega^{ab} + \omega^a_{\ c} \wedge \omega^{cb}$ and $T^a= \mathcal{D}^\omega e^a= de^a + \omega^a_{\ c} \wedge e^{c}$.\\

A new topological invariant can be added to the action of gravity, without affecting the classical equation of motions. The Holst term can be added to the Einstein-Hilbert action, then leading to the new action that involves a real (Barbero-Immirzi) parameter $\gamma$, i.e.: 
\begin{eqnarray}
S_{\rm EH}=\frac{1}{64 \pi G} \! \int \! d^4x  \left(  \epsilon_{ab}^{\ \ cd} \! +\! \frac{1}{\gamma} \delta_{ab}^{\ \ cd} \right) \! R_{\mu \nu}^{ \ \ \  ab} \, e_{\rho \, c} \,  e_{\sigma \, d}  \ \epsilon^{\mu \nu \rho \sigma} 
\!-\!  \frac{\Lambda}{3} \, \epsilon_{abcd} \,  \, e_{\mu}^{ \  a} \,  e_{\nu}^{ \ b} \, e_{\rho}^{ \ c} \,  e_{\sigma}^{ \ d}  \epsilon^{\mu \nu \rho \sigma}, \,
\end{eqnarray}
the phase-space of which retains the symplectic form: 
\begin{eqnarray}
\left\{ \!\!\!\phantom{a}^{\gamma}\omega^a_i(x),  \, \mathcal{E}^i_a(y)\right\}= \gamma\, \delta(x-y) \, \delta^j_i \delta^a_b\,,
\end{eqnarray}
where now $a,b=1,2,3$, the connection reads $\!\!\!\phantom{a}^{\gamma}\omega^a_i=\omega^{0a}_i+\frac{\gamma}{2} \epsilon^{0abc}\omega_{i\ bc}$, the (Plebanski) two-form reads $\mathcal{E}^i_a=\frac{4}{G} \epsilon_{abc}  \epsilon^{ijk} \, e^b_j\,e^c_k $, and as usual $\{\,,\,\}$ denote the Poisson brackets.

\subsection{BF formulation of gravity}

The Einstein-Hilbert-Holst action admits a formulation within the  BF framework, as a deviation from the topological theory. The BF theory is defined as a $G$-principle bundle on a D-dimensional base manifold $\mathcal{M}_D$. The action is the Killing form contraction of the $\mathfrak{g}$ Lie algebra-valued (D-2)-form $B$ and the field strength of the $G$-connection $A$. The Lagrangian then reads:
\begin{eqnarray}
\mathcal{L}_{\rm BF}=  {\rm tr} \left( B \wedge F[A]  \right) \,,
\end{eqnarray}
which specialized to the case of $SO(3,1)$ casts: 
\begin{eqnarray}
\mathcal{L}_{\rm BF}^{\rm SO(3,1)}=  B^{ab} \wedge F^{ab} [A] \,,
\end{eqnarray}
where here $a,b=1,\dots4$ are indices in the fundamental representation of $SO(3,1)$. This automatically provides the Einstein-Hilbert action, when the two-form is constrained to appear as a bi-vector, i.e. $B^{ab}=\epsilon^{ab}_{\ \ cd} \,\,  e^c\wedge e^d$. 
The Einstein-Hilbert action for gravity, complemented with the Holst term, is then instantiated by the imposition of the so-called simplicity constraints on the $B$ Lie algebra-valued two form, namely: 
$$
B^{ab}=\pm \left( \epsilon^{ab}_{\ \ cd} + \frac{1}{\gamma } \delta^{ab}_{\ \ cd} \right) \,\,  e^c\wedge e^d\,.
$$

\subsection{MacDowell-Mansouri action}

Switching now to the MacDowell-Mansouri action, we introduce an extended (anti-de Sitter) group $SO(3,2)$. In the MacDowell-Mansouri action, this is explicitly broken down to its stabilizer, the Lorenz group $SO(3,1)$. The anti-de Sitter connections, composed by ten internal components, are labelled by indices of the fundamental representation of the extended group $A,B=1,2,\dots5$ as $A^{AB}=A^{AB}_\mu dx^\mu$. This decomposes into $A^{AB}_\mu=(A^{ab}_\mu,\, A^{a5}_\mu)$, with $A^{ab}_\mu=\omega^{ab}_\mu$ and $ A^{a5}_\mu=\ell^{-1} \, e^a)$, given the identification: 
\begin{equation}
\frac{\Lambda}{3}=-\frac{1}{\ell^2}    \nonumber\,.
\end{equation} 
Involving the $\mathfrak{so}(3,2)$ algebra-valued connections, the decomposition is recognized to encode both the generators of the Lorentz transformations $M_{ab}$ and the space-time translation $P_a$, i.e.:
\begin{equation}
A_\mu= \frac{1}{2} \omega^{ab}_\mu M_{ab} + \frac{1}{\ell}  e^a_\mu P^a =\frac{1}{2} A_\mu^{AB} M_{AB}   \nonumber\,,
\end{equation} 
having identified $M_{a5}=P_a$.

According to this decomposition, the connection casts as:
\begin{eqnarray} \label{SBA}
A^{AB}= \left(\begin{array}{c c}
\omega^{ab} &  \frac{1}{\ell}  e^a\\
 -\frac{1}{\ell}  e^b & 0
\end{array}
\right) \,,
\end{eqnarray}
and correspondently the curvature 2-form, with indices contracted with the structure constants of the $SO(3,2)$ group:  
$$
F^{AB}=dA^{AB} + A^{AC} \wedge A_C^{\ B}\,,
$$
decomposes into the $SO(3,1)$ valued components: 
\begin{eqnarray} \label{compF}
F^{AB}= \left(\begin{array}{c c}
R^{ab} + \frac{1}{\ell^2}  e^a \wedge e^b &  \frac{1}{\ell}  T^a\\
 -\frac{1}{\ell}  T^b & 0
\end{array}
\right) \,.
\end{eqnarray}
The MacDowell-Mansouri action deploys this extended formalism, but with the crucial underlying assumption of  (explicit) symmetry breaking: 
\begin{eqnarray} \label{SB}
F^{AB} \rightarrow \bar{F}^{AB}=F^{ab}\,.
\end{eqnarray}
The Einstein-Hilbert action of gravity can then be encoded in a general framework, moving from the action: 
\begin{eqnarray}\label{MMa}
S_{\rm MM}[A] =\frac{\ell^2}{64 \pi G} \int {\rm tr} \left( \bar{F} \wedge \star F  \right)\,,
\end{eqnarray}
where $\star$ denotes the gravitational Hodge dual. Using the curvature decomposition in Eq.~\eqref{compF}, the action recasts as: 
\begin{eqnarray}
S_{\rm MM}[A] =\frac{\ell^2}{64 \pi G} \int \left( R^{ab} + \frac{1}{\ell^2}  e^a \wedge e^b \right) \wedge 
\left( R^{cd} + \frac{1}{\ell^2}  e^c \wedge e^d \right) \epsilon_{abcd} \,,
\end{eqnarray}
The action then entails the Einstein-Hilbert action, with the cosmological term, and some 4D Euler characteristic: 
\begin{eqnarray}\nonumber
32 \pi G \, S_{\rm MM}[A] =S_{\rm EH} + \frac{1}{2 \ell^2} \int  \epsilon_{abcd} \, e_{\mu}^{ \  a} \,  e_{\nu}^{ \ b} \, e_{\rho}^{ \ c} \,  e_{\sigma}^{ \ d}  \, \epsilon^{\mu \nu \rho \sigma} + \frac{\ell^2}{2}  \, \int \,\epsilon_{abcd}\, R_{\mu \nu}^{\ \ \ ab} R_{\rho \sigma}^{\ \ \ cd} \  \epsilon^{\mu \nu \rho \sigma}
\,.
\end{eqnarray}
The equations of motion read: 
\begin{eqnarray}
\left( R^{ab}\wedge e^c+ \frac{1}{2 \ell^2} e^a \wedge e^b \wedge e^c \right) \epsilon_{abcd} =0\,,
\qquad T^a=0
\,.
\end{eqnarray}

The MacDowell-Mansouri theory admits a straightforward $BF$ formulation, involving $\mathfrak{so}(2,3)$-valued $B$ two-forms. The action then reads: 
\begin{eqnarray}
S=\int_M {\rm tr} \left( B \wedge F -\frac{G \Lambda}{6} \bar{B} \wedge \star \bar{B} \right)\,,
\end{eqnarray}
the equations of motion of which imply that: i) varying in $\delta A^{AB}$, the local $\mathfrak{so}(2,3)$ Gauss constraint holds; ii)  varying in $\delta B_{a5}$, torsion vanishes, i.e. $F^{a5}=T^a \ell^{-1}=0$; iii) varying in $\delta \bar{B}_{ab}$, the relations that provide the MacDowell-Mansouri in Eq.~\eqref{MMa} is recovered, namely $F^{ab}=G\Lambda/3 \epsilon^{abcd} B_{cd}$. This relation allows us to write the MacDowell-Mansouri action as the deformation of a topological gauge theory. The symmetry breaking, which is here explicit, occurs as regulated by a coefficient that is dimensionless, in natural units, and for current estimates of the cosmological constant value reads $G\Lambda\sim 10^{-68}$. This makes General Relativity adapt to be described, with excellent approximation, as the perturbative limit of a topological field theory. \\

Without entering into further details, we notice that the MacDowell-Mansouri theory can be cast in a similar fashion in terms of an internal de Sitter group $SO(4,1)$, again explicitly broken down to $SO(3,1)$. Within this latter case, the $SO(4,1)$-connection is decomposed into the generator of translations and the generators of rotations, namely the tetrads and the spin-connection $A^{AB}=(\omega^{ab}, \ell^{-1} \, e^a)$.

\section{Wilczek gravity}
\label{5}
Frank Wilczek proposed in \cite{wilczek} a model that is reminiscent of the theory formulated by MacDowell and Mansouri, with internal $SO(4,1)$, or equivalently $SO(3,2)$, gauge symmetry. The configuration variables are the gauge symmetry connection $A^{AB}_\alpha$, and the internal scalar field $\phi^A$, in the fundamental representation of the group. The crucial difference between the MacDowell-Mansouri model and the Wilczek model lies in the spontaneous symmetry breaking of the internal gauge group that is present in the latter. This indeed directly enables us to recover the metric structure of General Relativity from the principal bundle of the model proposed by Wilczek.\\

The Lagrangian considered in \cite{wilczek} is:
\begin{equation}\label{W}
\begin{split} 
\mathcal{L}_{\rm W}=
&\kappa_3\, \epsilon^{\alpha \beta \gamma \delta} \epsilon_{ABCDE} \, F^{AB}_{\alpha \beta} \, \nabla_\gamma \phi^C\, \nabla_\delta \phi^D \, \phi^E\,,
\end{split}
\end{equation}
in which $\nabla_\gamma \phi^C= \partial_\gamma \phi^C + A^C_{\gamma F}\phi^F$ denotes the $SO(4,1)$ gauge covariant derivative.
The field strength is defined as: 
\begin{eqnarray}
F^{AB}_{\alpha\beta}= \partial_\alpha A^{AB}_\beta -\partial_\beta A^{AB}_\alpha + f^{AB}_{\ \ \ CDEF}A^{CD}_\alpha A^{EF}_\beta,
\end{eqnarray}
in which $f^{ABCDEF}$ is the structure constant of SO(4,1), namely: 
\begin{eqnarray} 
f^{AB\ LM\ PQ} &=& \eta^{BL} \eta^{AP} \eta^{MQ} - \eta^{AL} \eta^{BP} \eta^{MQ} \nonumber\\ 
&-& \eta^{BM} \eta^{AP} \eta^{LQ} + \eta^{AM}  \eta^{BP} \eta^{LQ}\,.          
\end{eqnarray} 

Two novel terms with respect to the MacDowell-Mansouri action, were introduced in the Wilczek model:\\
 
\begin{enumerate}
\item 
the interaction potential of $\phi^A$, namely:
$$\mathcal{L}_{1}= \kappa_1 \left(\eta_{AB} \phi^A \phi^B -v^2 \right)^2\,.$$
By varying with respect to $\phi^A$, this term is recognized to be stationarized either for $\phi^A=0$ or when $|\phi|=v$. In the latter case, the choice $\phi^A=\delta^A_5 v$ implements the spontaneous symmetry breaking.\\

\item 
a term that constrains the determinant of the metric in the spontaneous symmetry broken phase: 
\begin{equation}\nonumber 
\mathcal{L}_2= \kappa_2 \biggl( J-\omega \biggl)^2 \,,
\end{equation}
in which: 
\begin{equation}\label{J}
J= \epsilon^{\alpha\beta\gamma\delta}\epsilon_{ABCDE}\,\phi^E \nabla_\alpha \phi^A\nabla_\beta \phi^B \nabla_\gamma \phi^C \nabla_\delta \phi^D\,.
\end{equation}
This term is stationarized when $J=\omega$, implying the unimodularity of gravity \cite{Unimodular}. The spontaneous symmetry breaking induces in \eqref{J} a reduced expression for $J$, namely:  
\begin{equation}
J= v^5 \epsilon^{\alpha\beta\gamma\delta}\epsilon_{abcd}A^{a5}_\alpha A^{b5}_\beta A^{c5}_\gamma A^{d5}_\delta\, ,
\end{equation}
$J$ denoting the determinant of the metric.\\
\end{enumerate}

The total Lagrangian proposed by Wilczek then reads: 
\begin{equation}\nonumber
\begin{split}
\mathcal{L}_W= \kappa_2 \biggl( J-\omega \biggl)^2 +\kappa_1 \left(\eta_{AB} \Phi^A \Phi^B -v^2 \right)^2  +
\kappa_3\, \epsilon^{\alpha \beta \gamma \delta} \epsilon_{ABCDE} \, F^{AB}_{\alpha \beta} \, \nabla_\gamma \phi^C\, \nabla_\delta \phi^D \, \phi^E.
\end{split}
\end{equation}

In order to unveil the emergence of the gravity, we may instantiate the spontaneous symmetry breaking Eq.~\eqref{SB} directly on the Lagrangian in Eq.~\eqref{W}, using the decompositions recovered in Eq.~\eqref{SBA} and \eqref{compF}, and then finding: 
\begin{eqnarray}
\mathcal{L}_{\rm W}\!\!&=&\!\! \kappa_3\, v^3\,  \epsilon^{\alpha\beta\gamma\delta}\epsilon_{abcd}\biggl[ \bigl(\partial_\alpha \omega^{ab}_\beta - \partial_\beta\omega^{ab}_\beta + f^{ab}_{cdef}\omega^{cd}_\alpha \omega^{ef}_\beta \bigl) - \Lambda \, e^a_\alpha e^b_\beta \biggl] e^c_\gamma e^d_\delta \\ 
\!\!&=&\!\! \kappa_3\,  v^3\, \epsilon^{\alpha\beta\gamma\delta}\epsilon_{abcd}\biggl[ F^{ab}_{\alpha\beta} - \Lambda \, e^a_\alpha e^b_\beta \biggl] e^c_\gamma e^d_\delta.  \nonumber
\end{eqnarray}
This equation corresponds to the Einstein-Hilbert action, introduced in Eq.~\eqref{topological}, plus the cosmological constant term. The unimodular term is not essential for our arguments, and we can neglect it here. 

We can recast the main term of the Lagrangian density proposed by Wilczek, rearranging as:
\begin{eqnarray} \label{LLL2}
\mathcal{L}\!=\!\kappa \,  \epsilon^{\alpha\beta\gamma\delta}F_{\alpha\beta}^{AB}B_{\gamma\delta}^{AB} ,
\end{eqnarray}
where 
$$B_{\gamma\delta}^{AB}= \nabla_{\gamma} \Phi^{C} \nabla_{\delta}\Phi^{D}\, \Phi^{E}\, \epsilon^{AB}_{\ \ \ \ CDE}$$
works as a simplicity constraint, which here drags the Wilczek model away from its topological phase. 

In the Higgs condensate phase, the $B\wedge F$ term is (nothing but) reduced to the Einstein-Hilbert term, and an emergent diffeomorphism invariance is recovered starting from a topological invariance, which is finally broken. In this way, moving from a background independent theory, because of the flatness of the connection, after the Higgs condensate phase an emergent metric tensor is obtained. \\

\section{Conclusions and outlooks}
\label{6}
We have shown here that a generalized version of the holographic principle can be derived from fundamental considerations of quantum information theory, in particular, the imposition of separability on a joint state.  This GHP entails gauge invariance.  We emphasized that as soon as this is instantiated in an ambient Lorentzian space-time, gauge invariance under the Poincar\'e group automatically follows.  Indeed, following this pathway we can recover the action of gravity.  We summarize several gauge-invariant models for gravity, including gravity cast \`a la Wilczek. This is a formulation of the Einstein theory of gravity similar to the one proposed by MacDowell and Mansouri, which involves the representation theory of the Lie groups SO$(3,2)$ and SO$(4,1)$. \\

As the GHP provides a natural and completely general distinction between ``bulk'' and ``boundary'' degrees of freedom, one that is independent of geometry, it would be worth to investigate whether the AdS/CFT and dS/CFT correspondences could fit within this framework. This would require complementing the GHP with the concept of the renormalization group flow. Indeed, group renormalization flow techniques might be actually considered to connect the fully symmetric SO$(3,2)$ and SO$(4,1)$ theory with the SO(3,1) broken symmetric phase.

\end{document}